\documentclass[11pt]{article}

\usepackage[utf8]{inputenc}
\usepackage{authblk}
\usepackage{setspace}
\usepackage[margin=1.25in]{geometry}
\usepackage{graphicx}
\graphicspath{ {./figures/} }
\usepackage{subcaption}
\usepackage{amsmath}
\usepackage{booktabs}
\usepackage{wrapfig}
\usepackage{amssymb}
\usepackage{mathtools}
\usepackage{amsthm}
\usepackage{url}
\usepackage[numbers]{natbib}
\usepackage{xcolor}
\usepackage{hyperref}
\usepackage{bbm}
\usepackage{nicefrac}
\usepackage{multirow}
\usepackage{multicol}
\usepackage{dsfont}
\usepackage{amsfonts}
\usepackage{adjustbox}

\usepackage{amssymb}

\providecommand{\keywords}[1]
{
  \small	
  \textbf{\textit{Keywords---}} #1
}

\newtheorem{theorem}{Theorem}

\newtheorem{corollary}{Corollary}

\theoremstyle{definition}

\newtheorem{assumption}{Assumption}
\theoremstyle{remark}

\newcommand{\norm}[1]{\left\lVert#1\right\rVert}
\newcommand{\abs}[1]{\left\vert#1\right\vert}

\newcommand{\R}{\mathbb R}

\renewcommand{\thefootnote}{\fnsymbol{footnote}}

\def\E{\mathbb{E}}

\newcommand{\Cc}{\mathcal{C}}
\newcommand{\Dc}{\mathcal{D}}
\newcommand{\Ec}{\mathcal{E}}
\newcommand{\Fc}{\mathcal{F}}

\def\1v{\mathbf 1}
\def\0v{\mathbf 0}

\def\1{\mathbbm{1}}
\renewcommand{\P}{\mathbb{P}}
\newcommand{\ind}{\mathds{1}}

\title{Quickest Change-point Detection in Generalized Linear Models with Variational Estimation}

\author[1]{Jonghyeok Lee}
\author[1]{Yao Xie$^*$}

\affil[1]{H. Milton Stewart School of Industrial and Systems Engineering, Georgia Institute of Technology}

\begin{document}

\maketitle

\begingroup
\renewcommand{\thefootnote}{\fnsymbol{footnote}}
\footnotetext[1]{Correspondence: yao.xie@isye.gatech.edu}
\endgroup

\begin{abstract}
We consider quickest change-point detection in generalized linear models when the post-change parameter is unknown. Such problems arise, for example, when monitoring whether the relationship between a learned representation and an observed response changes over time. Since detection must occur soon after a change, the post-change parameter must be estimated from a limited number of observations, making finite-sample estimation accuracy particularly important. We propose a window-limited CUSUM procedure that estimates the post-change parameter using a variational-inequality (VI) formulation. For noncanonical links, VI estimation can offer favorable finite-sample accuracy over short windows despite the asymptotic efficiency advantage of maximum likelihood estimation. We characterize how estimation error reduces the effective post-change information rate and affects detection delay, and derive finite-window guarantees accounting for both statistical error and incomplete numerical optimization. The resulting procedure retains false-alarm control and achieves first-order asymptotic optimality under suitable conditions. Numerical experiments compare detection delay and update cost for fixed-budget estimation and benchmark VI-WLCUSUM against alternative detectors.
\end{abstract}


\section{Introduction}
\label{sec:intro}

Monitoring the relationship between a feature vector and an observed response is a basic problem in statistical signal processing. Examples include detecting a change in the calibration of a learned representation or in the response of a monitored system. A generalized linear model (GLM)\citep{mccullagh1989} provides a flexible conditional model for such relationships, but the parameter after a change is typically unknown.

Such changes arise in several modern applications. For example, learned language-model representations are mapped to next-token prediction probabilities through a softmax output layer \citep{vaswani2017}, and the calibration of these predictive distributions is an important measure of reliability \citep{jiang2021}. A change in the relationship between the learned representation and the observed response may therefore indicate a change in the underlying operating regime. Another example is sequential pairwise-choice data, where the Bradley--Terry model \citep{bradley1952} provides a logistic model for comparison outcomes, and changes in its parameters correspond to shifts in preferences or rankings. Change-point detection for sequential pairwise comparisons has been studied, for example, in~\cite{li2022}.

When the post-change distribution is unknown, quickest change detection must be coupled with parameter estimation. In this setting, parameter estimation is based on a limited number of post-change observations, since waiting for a large sample increases detection delay. Thus, the estimation problem arising in quickest detection is inherently finite-sample, and asymptotic efficiency alone may not characterize the estimator most suitable for detection.

Page's CUSUM \citep{page1954} is minimax optimal under Lorden's criterion when both the pre- and post-change distributions are known \citep{moustakides1986}. When the post-change parameter is unknown, window-limited CUSUM (WLCUSUM) \citep{xie2023} estimates it from a recent window and substitutes the estimate into the likelihood-ratio statistic. The quality of this short-window estimate therefore directly affects the post-change information available for detecting the change. This places particular importance on finite-sample estimation accuracy in addition to asymptotic efficiency.

In this work, we use variational-inequality (VI) estimation for the post-change GLM parameter \citep{zhu2026,juditsky2019}. VI estimation solves a moment equation whose geometry can be simpler than that of likelihood optimization. For canonical links, the VI and likelihood estimating equations coincide up to a known dispersion factor. For noncanonical links, they differ: maximum likelihood has the usual asymptotic efficiency advantage, whereas VI estimation can achieve better accuracy over the short windows relevant to quickest detection. This finite-sample behavior motivates incorporating VI estimation into WLCUSUM. Related VI estimators have also been studied for interacting marked processes \citep{juditsky2020}.

We make the connection between finite-window estimation and sequential detection explicit. In particular, we show that estimation error reduces the effective post-change information rate through the expected conditional KL discrepancy between the true and estimated post-change models. We derive finite-window VI risk bounds that account for both statistical estimation error and incomplete numerical optimization, establish false-alarm and detection-delay guarantees, and give conditions for first-order asymptotic optimality. Numerical experiments examine the resulting tradeoff among finite-sample estimation accuracy, computation, and detection delay.

\section{Setup}
\label{sec:model}

We observe a sequence of independent observations $Z_t = (X_t,Y_t)$, $t = 1, 2, \ldots$, with covariates $X_t \in \R^d$ following a common distribution $Q$. The response follows a standard Generalized Linear Model (GLM)
\[
    \E_{\theta}[Y \mid X] = h(X^\top \theta),
\]
where $h$ is a known continuous and nondecreasing (inverse) link function. We also allow clipped-link extensions.

We are interested in detecting the change in the parameter values, at an unknown time $\kappa \ge 0$ as soon as possible after it has occurred. This can be cast as the following sequential hypothesis test:
\begin{align*}
    H_0: \quad
    & Y_t \mid X_t \sim p_{\theta_\infty}(\cdot \mid X_t), \quad t = 1,2,\ldots,\\
    H_1: \quad
    & Y_t \mid X_t \sim
    \begin{cases}
        p_{\theta_\infty}(\cdot \mid X_t), \quad t = 1, 2, \ldots, \kappa,\\
        p_{\theta_0}(\cdot \mid X_t), \quad t = \kappa+1, \kappa+2, \ldots.
    \end{cases}
\end{align*}
The conditional density family and baseline parameter $\theta_\infty$ are known, while $\theta_0 \in \Theta \subset \R^d$ is unknown. We assume that the known parameter family $\Theta$ is compact and convex for the risk bounds, and $\Theta$ need not contain $\theta_\infty$.

Let $\Fc_t = \sigma(Z_1, \ldots, Z_t)$, with $\Fc_0$ trivial, and let $\P_\nu, \E_\nu$ denote the law and expectation for a change after $\nu$. The subscript $\infty$ denotes no change, whereas $0$ denotes an entirely post-change stream.
For a stopping time $T$, we seek small Lorden-type worst case delay \citep{lorden1971} subject to a large average run length (ARL) $\E_\infty T$:
\[
    \Dc(T) = \sup_{\nu \geq 0} \operatorname*{ess\,sup}
    \E_\nu[(T - \nu)^+ \mid \Fc_\nu],
\]
where $x^+ = \max\{x, 0\}$.

\section{Adaptive-CUSUM based on VI estimation}

Choose a window length $w \ge 1$. At time $t > w$, estimate the post-change coefficient from the preceding $w$ pairs through
\begin{equation}\label{eq:operator}
    F_{t,w}(\theta) = \frac{1}{w} \sum_{i=t-w}^{t-1} X_i \{h(X_i^\top \theta) - Y_i\}.
\end{equation}
On a fully post-change window, this mean-residual operator is centered at $\theta_0$. We define $\widehat\theta_{t-1} \in \Theta$ as an exact or approximate solution of the variational inequality (VI)
\begin{equation}\label{eq:VI}
    \langle F_{t,w}(\widehat\theta_{t-1}), \theta - \widehat\theta_{t-1} \rangle
    \ge -\epsilon_{t,w} \|\theta - \widehat\theta_{t-1}\|,
    \quad \theta \in \Theta,
\end{equation}
where $\epsilon_{t,w} \ge 0$ measures numerical error; $\epsilon_{t,w} = 0$ gives an exact solution.
This formulation also permits boundary solutions for which $F_{t,w}(\widehat\theta_{t-1})$ need not vanish.

The nondecreasing inverse link makes $F_{t,w}$ the gradient of the convex potential function
\begin{equation}\label{eq:potential}
    \Phi_{t,w}(\theta) = \frac{1}{w} \sum_{i=t-w}^{t-1} \{H(X_i^\top \theta) - Y_i X_i^\top \theta\}, \quad H'=h.
\end{equation}
Thus an exact VI solution minimizes $\Phi_{t,w}$ over $\Theta$. For example, with Euclidean projection $\Pi_\Theta$, positive step sizes $\eta_k$, and iteration budget $K_{\rm it}$,
\[
    \theta^{(k+1)} = \Pi_\Theta \{\theta^{(k)} - \eta_k F_{t,w}(\theta^{(k)})\},
    \quad \widehat\theta_{t-1} = \theta^{(K_{\rm it})}.
\]
The iterations run from $k=0$ to $K_{\rm it} - 1$, starting in $\Theta$.

The resulting VI-WLCUSUM uses threshold $b > 0$, starts at $S_w = 0$,
and, for $t > w$, computes
\begin{align}\label{eq:detector}
    \ell_t = \log\frac{p_{\widehat\theta_{t-1}}(Y_t \mid X_t)}{p_{\theta_\infty}(Y_t \mid X_t)}, \quad
    S_t = \max\{S_{t-1},0\} + \ell_t,
\end{align}
and raises an alarm at the stopping time
\begin{equation}\label{eq:st-time}
    T_b = \inf\{t > w: S_t \ge b\}.
\end{equation}
The estimate is $\Fc_{t-1}$-measurable as the current pair is scored before entering a later estimation window. This separation is essential for false-alarm control.
The recursion discards a negative accumulated score before adding the new increment and stops at the first threshold crossing. It is the WLCUSUM recursion of~\cite{xie2023}, with VI supplying the estimate.

\section{Theoretical Guarantees}
\label{sec:core-theory}

We first study the false-alarm control coming from normalized predictive densities. For delay, we first bound the VI estimation error, translate it into lost post-change information, and then account for overlapping windows and threshold overshoot. Throughout, $w$ and $b$ are fixed unless an asymptotic limit is specified.
Please see Appendix~\ref{app:proof} for the proofs.

\begin{assumption}[Model and likelihood]\label{ass:core-model}
All fitted densities are measurable, and have a common support with the baseline for $Q$-almost every $X$. Define
\begin{equation}\label{eq:core-information}
\begin{aligned}
    K(\theta_0, \theta) &= \E_Q D_{\mathrm{KL}} \{p_{\theta_0}(\cdot \mid X)\|p_\theta(\cdot \mid X)\}, \quad
    I = K(\theta_0, \theta_\infty) \in (0,\infty).
\end{aligned}
\end{equation}
\end{assumption}
\noindent Here $D_{\mathrm{KL}}$ is Kullback--Leibler divergence.
Thus $I$ is the mean log-likelihood increment available to an oracle knowing $\theta_0$, while $K(\theta_0,\theta)$ measures the predictive cost of using a fitted coefficient $\theta$.

\begin{assumption}[Fixed-window estimation]\label{ass:core-window}
For each $w$, one deterministic measurable map $g_w$ is used at every time:
\[
    \widehat\theta_{t-1} = g_w(Z_{t-w}, \ldots, Z_{t-1}), \quad t>w.
\]
\end{assumption}
\noindent Assumption~\ref{ass:core-window} ensures that the increment term depends on only $w+1$ consecutive observations.
On an i.i.d.~post-change window, write $\widehat\theta_w$ for its fit and $F_w(\theta) = w^{-1} \sum_{i=1}^w X_i \{h(X_i^\top\theta) - Y_i\}$.

\begin{assumption}[VI accuracy and information loss]\label{ass:core-vi}
The following hold for this post-change window.
\begin{enumerate}
\item[(i)] The feasible set $\Theta$
has diameter $D$. Moreover,
\[
    \Gamma = \E_0 [(Y - h(X^\top \theta_0))^2 X X^\top],
    \quad \operatorname{tr} \Gamma < \infty.
\]
\item[(ii)] On an event $G_w$ with $\P_0(G_w^c) \le \pi_w \in [0,1]$, for deterministic $m > 0$,
\[
    \langle F_w(\theta) - F_w(\theta_0), \theta - \theta_0 \rangle
    \ge m \|\theta - \theta_0\|^2, \quad \theta \in \Theta.
\]
\item[(iii)] The fit $\widehat\theta_w \in \Theta$ satisfies
\[
    \langle F_w(\widehat\theta_w), \theta - \widehat\theta_w \rangle
    \ge -\epsilon_w \|\theta - \widehat\theta_w\|, \quad \theta \in \Theta,
\]
where $\epsilon_w \ge 0$ may depend on the data and $e_w = \sqrt{\E_0 \epsilon_w^2} < \infty$.
\item[(iv)] For $0 \le C_K < \infty$, $K(\theta_0, \theta) \le C_K \|\theta - \theta_0\|^2$ on $\Theta$.
\end{enumerate}
\end{assumption}
\noindent Assumption~\ref{ass:core-vi} requires informative windows and controlled numerical error. For differentiable $h$, uniform positive lower bounds on $h'$ over feasible predictors and on the sample-Gram minimum eigenvalue yield the curvature condition (ii) \citep[Eq.~(18)]{juditsky2023}; (iv) separately
converts parameter error into predictive KL loss.

\begin{assumption}[Predictive moment]
\label{ass:core-moment}
In the stationary post-change regime, the predictive increment in~\eqref{eq:detector} satisfies
$\E_0 \ell_t^2 \le M_{2,w} < \infty$.
\end{assumption}
\noindent This moment condition controls threshold overshoot. Now define
\begin{equation}\label{eq:core-error}
\begin{aligned}
    \Ec_w &= \frac{(\sqrt{\operatorname{tr} \Gamma/w} + e_w)^2}{m^2} + D^2 \pi_w, \quad
    \underline I_w = I - C_K \Ec_w.
\end{aligned}
\end{equation}
The quantity $\Ec_w$, used below to bound the VI estimation error, accounts for sampling variability through $\operatorname{tr}\Gamma/w$, numerical inaccuracy through $e_w$, and poorly informative windows through $\pi_w$.
\begin{theorem}[False-alarm guarantee]
\label{thm:core-arl}
Under Assumption~\ref{ass:core-model}, any predictable estimates yielding valid densities satisfy
\begin{equation}\label{eq:core-arl}
    \E_\infty T_b \ge w + e^b.
\end{equation}
\end{theorem}

Theorem~\ref{thm:core-arl} separates false-alarm control from estimation accuracy. The additive $w$ reflects our convention that $T_b$ includes the initial estimation window.

\begin{theorem}[VI error and effective information]
\label{thm:core-risk}
Under Assumptions~\ref{ass:core-model}--\ref{ass:core-vi},
\begin{equation}\label{eq:core-risk-drift}
\begin{aligned}
    \E_0 \|\widehat\theta_w - \theta_0\|^2 &\le \Ec_w, \quad
    I_w \coloneqq \E_0 \ell_t
    = I - \E_0 K(\theta_0, \widehat\theta_w)
    \ge \underline I_w,
\end{aligned}
\end{equation}
where the increment is evaluated in the stationary post-change regime.
\end{theorem}

The identity in~\eqref{eq:core-risk-drift} shows that the information lost through estimation is exactly the expected predictive KL divergence. Under Assumption~\ref{ass:core-vi}(iv), parameter mean squared error controls this loss of mean post-change information.

\begin{theorem}[Worst-case detection delay]
\label{thm:core-delay}
Under Assumptions~\ref{ass:core-model}--\ref{ass:core-moment}, if $\underline I_w > 0$, then
\begin{equation}\label{eq:core-delay}
\begin{aligned}
    \Dc(T_b) \le
    \frac{b + wI}{\underline I_w}
    + \frac{M_{2,w}}{\underline I_w^2} + \sqrt{\frac{M_{2,w}(b + wI)}{\underline I_w^3}}.
\end{aligned}
\end{equation}
\end{theorem}

For fixed $w$, the leading term in~\eqref{eq:core-delay} as $b \to \infty$ is $b/\underline I_w$. The remaining terms account for the estimation window and threshold overshoot. A larger window may improve estimation and increase $\underline I_w$, but also enlarges the explicit $wI$ term in the numerator. The bound therefore captures the competing effects of window length.

\begin{corollary}[First-order minimax optimality]\label{cor:core-first-order}
Fix $d$ and $\theta_0$. Suppose Assumptions~\ref{ass:core-model}--\ref{ass:core-moment} hold along $w = w_\gamma$, with $m > 0$, $D$, $C_K$, and $\Gamma$ independent of $w$, and
\begin{gather}\label{eq:core-asymptotic-conditions}
    e_w \to 0, \quad \pi_w \to 0, \quad \sup_w M_{2,w} < \infty, \nonumber\\
    w \to \infty, \quad w = o(\log\gamma).
\end{gather}
For $b = \log\gamma$,
\begin{equation}\label{eq:core-first-order}
    \E_\infty T_b \ge \gamma, \qquad
    \Dc(T_b) = \frac{\log\gamma}{I} \{1+o(1)\},
\end{equation}
and $T_b$ is first-order minimax optimal under Lorden's criterion.
\end{corollary}

Corollary~\ref{cor:core-first-order} shows that VI-WLCUSUM attains the classical minimax benchmark $(\log\gamma)/I$ under the false-alarm constraint $\E_\infty T \ge \gamma$. This is also the first-order worst-case delay achieved by an oracle CUSUM that knows $\theta_0$. Thus, estimating the unknown post-change parameter and using a finite window incur no first-order asymptotic delay penalty.

The following examples attain the first-order minimax guarantee~\eqref{eq:core-first-order} for fixed dimension, parameters, and model constants, with $I > 0$, $b = \log\gamma$, $w \to \infty$, and $w = o(\log\gamma)$.

\subsection{Gaussian example}\label{sec:gaussian-example}

Let $X \sim N(0, I_d)$ and
\begin{equation}\label{eq:gaussian}
\begin{aligned}
    Y &= h(X^\top \theta) + \varepsilon, \quad \varepsilon \sim N(0, \sigma^2),\\
    h(z) &= az + (1-a) \tanh z, \quad 0 < a < 1,
\end{aligned}
\end{equation}
where $\varepsilon$ is independent of $X$ and $\sigma > 0$. Let $\Theta$ be a radius-$R$ ball containing $\theta_\infty$ and $\theta_0$. For $w \ge 16d$, use projected-gradient updates of $\Phi_{t,w}$ from the same fixed point in $\Theta$ at every window, with
\[
    \eta = \frac{8}{9+a}, \quad
    k_w = \left\lceil \frac{9}{4a} \log\frac{w}{d} \right\rceil, \quad
    \rho = \frac{9-a}{9+a}.
\]
Gaussian concentration~\cite[Cor.~5.35]{vershynin2012} and this iteration budget give
\begin{equation}\label{eq:explicit-risk}
\begin{aligned}
    \E_0 \|\widehat\theta_{w,k_w} - \theta_0\|^2 \le \frac{32 \sigma^2 d}{a^2 w} +8 R^2 \rho^{2k_w} + 8 R^2 e^{-w/32} = O(dw^{-1}).
\end{aligned}
\end{equation}
Here $C_K = 1/(2\sigma^2)$ and the predictive moment satisfies
\[
    \sup_w M_{2,w} \le \frac{4 R^2}{\sigma^2} + \frac{48 R^4}{\sigma^4}.
\]
Thus $I_w \to I > 0$, and the delay argument of Theorem~\ref{thm:core-delay}, together with the minimax lower bound, gives~\eqref{eq:core-first-order}.

\subsection{Clipped-Poisson example}\label{sec:clipped-example}

Let $Y \mid X \sim \operatorname{Pois}(h(X^\top \theta))$, where
\[
    h(z) = \max\{c, \min(e^z, C)\}, \qquad 0 < c < C.
\]
Assume $\|X\| \le B_X$ and $\sup_{\theta \in \Theta}|X^\top \theta| \le B$ almost surely, with $B > 0$. For some $\delta, g > 0$, require
\begin{equation}\label{eq:clipped-active}
    \E_Q \! \left [\ind_{\{\log c + \delta \le X^\top \theta_0 \le \log C - \delta\}} X X^\top \right] \succeq g I_d.
\end{equation}
This active-region condition and matrix Chernoff~\cite{tropp2012} give
\begin{equation}
\begin{aligned}
 m&=\frac{gc}{2}\min\{1,\delta/(2B)\},\\
 \pi_w&=\min\{1,d\exp[-wg/(8B_X^2)]\}.
\end{aligned}
\label{eq:clipped-risk}
\end{equation}
The bounded means also give $C_K = C^2 B_X^2 / (2c)$, $\operatorname{tr} \Gamma \le C B_X^2$, and $\sup_w M_{2,w} < \infty$. Hence a deterministic fixed-window VI rule with $e_w \to 0$ attains~\eqref{eq:core-first-order} by Corollary~\ref{cor:core-first-order}.

\section{Numerical experiments}
\label{sec:experiments}
We study the noncanonical softplus-Poisson model $Y_{t,j} \mid X_t \sim \operatorname{Pois}(h(X_t^\top\theta))$, $j=1, \ldots, q$, where $h(z) = \log(1 + e^z)$, the covariate coordinates are independent $\operatorname{Unif}(-\sqrt3, \sqrt3)$, and $\theta_\infty = 0$. Responses are conditionally independent.

\subsection{Fixed-budget VI versus likelihood estimation}
For $q = 1$ and $d \in \{20, 100, 200, 400\}$, set $\theta_0 = \delta_d u_d$, with $u_{d,1} = 1/2$ and $u_{d,k} = \sqrt{0.75/(d-1)}$ for $k \ge 2$. We calibrate $\delta_d$ to $I = 0.08$ using 200,000 covariates. VI and the likelihood-based WLCUSUM comparator~\cite{xie2023} use $w = 320$, 64 projected steps, and
\[
 \Theta_d = \{\theta: \|\theta\|_2 \le 1.5,\;
 \theta_1 \ge \delta_d/8\}.
\]
This constraint provides both methods with a positive-direction and minimum-shift restriction and excludes the null. Initial VI/likelihood step sizes are $0.00125/0.0025$, with per-fit decay $0.995$.
Thresholds are calibrated to target the ARL $\E_\infty T = 1000$.
Timing uses 40 paired paths and 320 updates per path.

\begin{table}[t]
\centering
\setlength{\tabcolsep}{1pt}
\renewcommand{\arraystretch}{1.0}
\caption{Scalar-response comparison. Post-window delay: mean (SE);
$\Delta D$: VI minus MLE, with paired 95\% CI. Update time in ms.}
\label{tab:softplus-matched}
\begin{tabular*}{\columnwidth}{
  @{\extracolsep{\fill}}
  rrrr
  @{\extracolsep{0pt}\hspace{6pt}}
  r
  @{\hspace{2pt}}
  r
  @{}
}
\toprule
& \multicolumn{2}{c}{Mean delay (SE)}
&& \multicolumn{2}{c}{Update}\\
\cmidrule(lr){2-3}\cmidrule(lr){5-6}
$d$ & VI & MLE & $\Delta D$ [95\% CI] & VI & MLE\\
\midrule
20 & 49.97 (1.15) & 49.65 (1.33)
   & 0.33 [\textminus{}0.87,1.33] & 2.07 & 4.31\\
100 & 102.72 (3.77) & 101.60 (4.06)
    & 1.13 [\textminus{}3.07,5.31] & 2.92 & 5.17\\
200 & 133.80 (4.89) & 140.72 (5.07)
    & \textminus{}6.93 [\textminus{}12.86,\textminus{}1.13]
    & 4.12 & 6.31\\
400 & 178.37 (6.83) & 182.64 (7.35)
    & \textminus{}4.27 [\textminus{}13.55,4.21]
    & 6.30 & 8.45\\
\bottomrule
\end{tabular*}
\end{table}

Table~\ref{tab:softplus-matched} shows that VI achieves lower update times across all tested dimensions while maintaining comparable detection delays. At $d=200$ and $400$, VI also yields smaller mean delays, although the uncertainty intervals do not establish a systematic high-dimensional advantage. The clearest benefit is computational efficiency under a fixed iteration budget. While MLE is asymptotically efficient under standard regularity conditions, this does not imply better detection performance with a finite estimation window and limited optimization. These results suggest that VI can provide competitive detection performance at lower update cost, with potential delay gains in higher-dimensional, finite-sample settings.

\subsection{VI-WLCUSUM versus alternative detectors: multivariate responses}
We use $d \in \{1,2\}$ and $q \in \{1,2,4\}$, with $\theta_0 = a_{d,q,I} \mathbf{1}_d / \sqrt{d}$. Magnitudes are calibrated to joint KL information $I \in \{0.02,0.04,0.08\}$, so increasing $q$ does not increase information per time step.
Here $\Theta = \{\theta: \|\theta\|_2 \le 1.5\}$
imposes no directional constraint. 

Let $N_t = \sum_{j=1}^q Y_{t,j}$ and $\mu_\infty = \log2$. VI-WLCUSUM uses $w = 32$,
64 steps, initial step $0.1$, and decay $0.995$.
Score-based procedures monitor $V_t = X_t (N_t - q \mu_\infty) / \sqrt{q \mu_\infty}$. The multivariate exponentially weighted moving average (MEWMA) chart~\cite{lowry1992multivariate,zhang2023concept} uses weight $2/33$ (variance-equivalent memory 32). A moving-window Hotelling-type $T^2$ chart~\cite{hotelling1947multivariate} uses $\|\sum_{i=t-31}^t V_i\|^2 / 32$, and two-sided coordinatewise CUSUM charts~\cite{page1954} use reference $0.25$. Multichart CUSUM~\cite{tartakovsky2005asymptotic} and mixture Shiryaev--Roberts (MSR)~\cite{tartakovsky2018asymptotic} use the 12/112 nonzero $0.25$-lattice points in the radius-1.5 ball for $d=1/2$.
Oracle CUSUM~\cite{page1954} knows $\theta_0$ and is a benchmark.

The change follows 256 null-training samples. VI and the score smoothers initialize from the last 32 training samples, while cumulative charts start at zero. Each method uses 12,000 null calibration and 10,000 separate validation paths, targeting ARL $\E_\infty T = 1000$.

\begin{table}[!t]
\centering
\setlength{\tabcolsep}{1pt}
\renewcommand{\arraystretch}{0.945}
\caption{Mean initialized delay over 2,000 paths at target ARL 1000. Bold: smallest non-oracle mean.
VI: VI-WLCUSUM; EW: score-based MEWMA; $T^2$: moving-window Hotelling-type chart; SC: score-based CUSUM; MC: multichart CUSUM; MSR: mixture Shiryaev--Roberts; OC$^\dagger$: oracle CUSUM.}
\label{tab:memory32-comparison}
\begin{tabular*}{\columnwidth}{@{\extracolsep{\fill}}rrrrrrrr@{\hspace{10pt}}r@{}}
\toprule
$q$ & $I$ & VI & EW & $T^2$ & SC & MC & MSR & OC$^\dagger$\\
\cmidrule(lr){1-8}\cmidrule(l){9-9}
\multicolumn{9}{c}{$d=1$}\\
\midrule
1 & 0.02 & \textbf{155.3} & 194.9 & 188.9 & 190.0 & 173.3 & 162.4 & 114.3\\
  & 0.04 & \textbf{86.3}  & 105.5 & 102.4 & 102.9 & 99.5  & 97.7  & 71.5\\
  & 0.08 & \textbf{49.3}  & 57.8  & 55.3  & 56.4  & 55.6  & 58.6  & 43.9\\
2 & 0.02 & \textbf{151.9} & 187.8 & 193.9 & 191.4 & 181.9 & 164.9 & 111.7\\
  & 0.04 & \textbf{82.5}  & 99.7  & 97.6  & 99.0  & 100.7 & 96.3  & 69.6\\
  & 0.08 & \textbf{48.8}  & 54.0  & 53.3  & 55.7  & 57.9  & 57.9  & 43.5\\
4 & 0.02 & \textbf{155.0} & 194.0 & 191.3 & 188.3 & 201.2 & 173.9 & 110.6\\
  & 0.04 & \textbf{85.0}  & 99.6  & 100.1 & 98.7  & 108.0 & 99.6  & 69.6\\
  & 0.08 & \textbf{47.3}  & 52.6  & 52.4  & 53.9  & 60.0  & 58.3  & 42.9\\
\midrule
\multicolumn{9}{c}{$d=2$}\\
\midrule
1 & 0.02 & 205.1          & 256.3 & 247.3 & 241.1 & 199.1 & \textbf{192.3} & 113.2\\
  & 0.04 & \textbf{108.6} & 132.0 & 129.4 & 132.6 & 113.1 & 113.5 & 70.3\\
  & 0.08 & \textbf{59.7}  & 70.8  & 67.0  & 73.3  & 62.7  & 65.2  & 43.6\\
2 & 0.02 & \textbf{202.6} & 244.1 & 247.0 & 249.5 & 215.8 & 208.6 & 108.9\\
  & 0.04 & \textbf{107.9} & 141.4 & 133.2 & 137.1 & 120.0 & 120.1 & 72.1\\
  & 0.08 & \textbf{59.7}  & 69.1  & 68.0  & 73.1  & 66.4  & 67.0  & 43.4\\
4 & 0.02 & \textbf{196.3} & 244.5 & 248.4 & 242.1 & 241.2 & 231.9 & 109.3\\
  & 0.04 & \textbf{106.6} & 132.1 & 131.4 & 133.1 & 129.4 & 129.1 & 70.9\\
  & 0.08 & \textbf{59.5}  & 68.8  & 68.2  & 73.9  & 70.8  & 71.8  & 43.5\\
\bottomrule
\end{tabular*}
\end{table}

Table~\ref{tab:memory32-comparison} shows that VI-WLCUSUM generally improves delay over the alternatives in these settings. VI has the smallest mean in most settings, but mixture-SR is better for the weakest change at $(d,q) = (2,1)$.
This exception is consistent with estimation uncertainty offsetting the benefit of adaptation under weak signals. Oracle CUSUM has shorter delays in all settings, indicating a remaining penalty for the unknown parameter.
The delay advantage comes with more computation than the classical detectors, especially the score-based methods. Thus VI saves computation relative to the likelihood fit in the first study, but VI-WLCUSUM spends more computation for improved delay relative to simple score charts in the second.

\section{Conclusion}\label{sec:conclusion}
We proposed VI-WLCUSUM for quickest change detection in GLMs with unknown post-change parameters. Our analysis relates statistical and optimization errors in VI estimation to the loss of post-change information, leading to finite-window delay bounds and first-order asymptotic optimality. Our formulation provides an alternative to likelihood fitting that can reduce the cost of updating the detector. The numerical results demonstrate this computational benefit while maintaining competitive detection performance.

The delay guarantees require a fixed-window estimator and sufficient empirical curvature, which need not hold for warm-start implementations or high-dimensional problems with short windows. Extending the analysis to these settings is an important next step. Adaptive choices of the window length may further improve the balance between estimation cost and detection delay.

\section*{Acknowledgment}
This work is supported by the Office of Naval Research
[N000142412278].

\section*{Compliance with Ethical Standards}
This is a numerical simulation study for which no ethical approval was required.

\bibliographystyle{IEEEbib}
\bibliography{references}

\appendix

\section{Proofs}\label{app:proof}

\begin{proof}[Proof of Theorem~\ref{thm:core-arl}]
    We follow the martingale argument in \cite[Lemma 3]{xie2023}. For $t > w$, let $\Lambda_t = e^{\ell_t}$. Let $\lambda$ be a common dominating measure for the conditional distributions from Assumption~\ref{ass:core-model}. Under $\P_\infty$, independence of $Z_t$ from $\Fc_{t-1}$ gives, for every measurable set $B$,
    \[
        \P_\infty(Y_t \in B \mid \Fc_{t-1}, \sigma(X_t)) = \int_B p_{\theta_\infty} (y \mid X_t) \lambda(dy) \quad \text{a.s.}
    \]
    Since $Y_t \in \operatorname{supp} p_{\theta_\infty} (\cdot \mid X_t)$ almost surely, Assumption~\ref{ass:core-model} gives
    \begin{align*}
        \E_\infty [\Lambda_t \mid \Fc_{t-1}, \sigma(X_t)]
        &= \int_{\operatorname{supp} p_{\theta_\infty} (\cdot \mid X_t)} \frac{p_{\widehat \theta_{t-1}}(y \mid X_t)}{p_{\theta_\infty} (y \mid X_t)} p_{\theta_\infty} (y \mid X_t) \lambda(dy)\\
        &= \int_{\operatorname{supp} p_{\theta_\infty} (\cdot \mid X_t)} p_{\widehat \theta_{t-1}}(y \mid X_t)  \lambda(dy) = 1 \quad \text{a.s.}
    \end{align*}
    The tower property yields
    \begin{align}\label{eq:pf-1-lr-tower}
        \E_\infty [\Lambda_t \mid \Fc_{t-1}] = 1 \quad \text{a.s.}
    \end{align}

    Define the nonnegative adapted process
    \[
        R_w = 0, \quad R_t = (1+R_{t-1}) \Lambda_t, \quad t>w.
    \]
    By~\eqref{eq:pf-1-lr-tower}, 
    \[
        \E_\infty[R_t \mid \Fc_{t-1}] = (1 + R_{t-1}) \E_\infty[\Lambda_t \mid \Fc_{t-1}] = 1 + R_{t-1}.
    \]
    Induction from $R_w = 0$ gives $\E_\infty R_t = 1 + \E_\infty R_{t-1} = t - w < \infty$.
    Thus $M_t = R_t - (t-w)$ is adapted and integrable, with $M_w = 0$ and
    \[
        \E_\infty [M_t \mid \Fc_{t-1}] = 1 + R_{t-1} - (t-w) = R_{t-1} - \{(t-1) - w\} = M_{t-1}.
    \]
    Hence $(M_t)_{t \ge w}$ is an $(\Fc_t)$-martingale.

    Note that $R_{w+1} = \Lambda_{w+1} = e^{S_{w + 1}}$. If $R_{t-1} \ge  e^{S_{t-1}}$ for $t \ge w+2$, then
    \[
        R_t = (1+R_{t-1}) \Lambda_t \ge \max\{1, e^{S_{t-1}}\} e^{\ell_t} = e^{\max\{0, S_{t-1}\} + \ell_t} = e^{S_t}.
    \]
    Consequently, by induction, $R_t \ge e^{S_t}$ for every $t > w$. For an integer $N > w$, bounded optional sampling at $T_b \wedge N$ yields $\E_\infty M_{T_b \wedge N} = \E_\infty M_w = 0$. Therefore, since $R_{T_b \wedge N} \ge e^{S_{T_b \wedge N}} \ge e^b \ind_{T_b \le N}$,
    \[
        \E_\infty [T_b \wedge N] - w = \E_\infty R_{T_b \wedge N} \ge e^b \P_\infty(T_b \le N).
    \]
    If $\E_\infty T_b = \infty$, the result is immediate. Otherwise $T_b < \infty$ almost surely, and as $N \to \infty$, monotone convergence gives $\E_\infty T_b \ge w + e^b$.
\end{proof}

\begin{proof}[Proof of Theorem~\ref{thm:core-risk}]
    \emph{Likelihood-ratio integrability.} For probability measures $P_1, P_2$ with densities $p_1, p_2$ relative to a common dominating measure and $D_{\mathrm{KL}}(P_1 \| P_2) <  \infty$, the standard bound \citep{barron1988exponential} gives
    \begin{equation}\label{eq:lr-bound}
        \E_{P_1} \abs{\log(p_1/p_2)} \le D_{\mathrm KL}(P_1 \| P_2) + 2/e < \infty.
    \end{equation}
    For $\theta \in \Theta$ and $Z = (X, Y)$, define
    \[
        \ell(\theta; Z) = \log \frac{p_\theta (Y \mid X)}{p_{\theta_\infty} (Y \mid X)}.
    \]
    Thus, $\ell_t = \ell(\widehat \theta_{t-1}; Z_t)$. Applying~\eqref{eq:lr-bound} with $(P_1, P_2) = (P_{\theta_0}, P_{\theta_\infty})$ and $(P_1, P_2) = (P_{\theta_0}, P_\theta)$, then using Assumption~\ref{ass:core-vi}(i), (iv) yields
    \begin{equation}\label{eq:uniform-lone}
    \begin{aligned}
        \sup_{\theta \in \Theta} \E_0 \abs{\ell(\theta; Z)}
        &\le \E_0 \abs{\log \frac{p_{\theta_0}(Y \mid X)}{p_{\theta_\infty}(Y \mid X)}} + \sup_{\theta \in \Theta} \E_0 \abs{\log \frac{p_{\theta_0}(Y \mid X)}{p_\theta(Y \mid X)}}\\
        &\le I + \sup_{\theta \in \Theta} K(\theta_0, \theta) + 4/e\\
        &\le I + C_K D^2 + 4/e < \infty.
    \end{aligned}
    \end{equation}

    \emph{VI risk and information loss.}
    We adapt the VI estimation error arguments of \cite[Prop.~3]{juditsky2019} and \cite[Prop.~4.2]{juditsky2023}. For the i.i.d.~post-change training window, put
    \[
        \psi_i = X_i \{h(X_i^\top \theta_0) - Y_i\}, \quad u_w = \widehat \theta_w - \theta_0.
    \]
    Correct mean specification and Assumption~\ref{ass:core-vi}(i) gives
    \[
        \E_0[\psi_i \mid X_i] = 0, \quad
        \E_0[\psi_i \psi_i^\top] = \Gamma.
    \]
    Independence therefore implies
    \[
        \E_0 \norm{F_w(\theta_0)}^2 = \frac{1}{w^2} \sum_{i=1}^w \E_0 \norm{\psi_i}^2 = \frac{\operatorname{tr}\Gamma}{w}.
    \]
    Since $\theta_0 \in \Theta$, the (approximate) VI~\eqref{eq:VI} at $\theta_0$ gives $\langle F_w(\widehat \theta_w), u_w \rangle \le \epsilon_w \norm{u_w}$. On $G_w$, Assumption~\ref{ass:core-vi}(ii) and Cauchy-Schwarz inequality then yield
    \begin{equation*}\label{eq:pathwise-risk}
    \begin{aligned}
        m \norm{u_w}^2
        \le \langle F_w(\widehat \theta_w) - F_w(\theta_0), u_w \rangle
        \le \epsilon_w \norm{u_w} - \langle F_w(\theta_0), u_w \rangle
        \le \{\epsilon_w + \norm{F_w(\theta_0)}\} \norm{u_w}.
    \end{aligned}
    \end{equation*}
    Thus, $\norm{u_w} \le m^{-1} \{\norm{F_w(\theta_0)} + \epsilon_w\}$ on $G_w$. On $G_w^c$, feasibility gives $\norm{u_w} \le D$. Therefore,
    \begin{equation}\label{eq:risk-final}
    \begin{aligned}
        \E_0 \norm{u_w}^2
        &\le \frac{1}{m^2}
        \E_0[(\norm{F_w(\theta_0)} + \epsilon_w)^2 \ind_{G_w}] + D^2 \pi_w\\
        &\le \frac{\{(\E_0 \norm{F_w(\theta_0)}^2)^{1/2} + e_w\}^2}{m^2} + D^2 \pi_w\\
        &= \frac{(\sqrt{\operatorname{tr} \Gamma/w} + e_w)^2}{m^2} + D^2 \pi_w = \Ec_w.
    \end{aligned}
    \end{equation}
    For a fixed feasible $\theta$,
    \[
        \ell(\theta; Z) = \log \frac{p_{\theta_0}(Y \mid X)}{p_{\theta_\infty}(Y \mid X)} - \log \frac{p_{\theta_0}(Y \mid X)}{p_\theta(Y \mid X)},
    \]
    gives
    \begin{equation}\label{eq:fixed-KL}
        \E_0 \ell(\theta; Z) = I - K(\theta_0, \theta).
    \end{equation}
    Let $Z=Z_{w+1}$, independent of $(Z_1,\ldots,Z_w)$ under $\P_0$. Therefore, independence and integrability in \eqref{eq:uniform-lone}, together with \eqref{eq:risk-final}, imply
    \begin{equation}\label{eq:information-final}
    \begin{aligned}
        I_w
        = \E_0 \left[\E_0 [\ell(\widehat \theta_w; Z) \mid \widehat \theta_w] \right]
        = I - \E_0 K(\theta_0, \widehat \theta_w)
        \ge I - C_K \E_0 \norm{u_w}^2 \ge I - C_K \Ec_w = \underline I_w.
    \end{aligned}
    \end{equation}
    In particular, on a post-change stream,
    \begin{equation}\label{eq:conditional-kl}
        \E_0 [\ell_t \mid \Fc_{t-1}]
        = I - K(\theta_0, \widehat \theta_{t-1})
        \le I,
    \end{equation}
    where the same bound~\eqref{eq:uniform-lone} guarantees that this conditional expectation is well-defined.
\end{proof}

\begin{proof}[Proof of Theorem~\ref{thm:core-delay}]
    Suppose all observations follow the post-change law. By the common window, we can write
    \[
        \ell_t = \phi_w(Z_{t-w}, \ldots, Z_t), \quad t>w,
    \]
    for one measurable function $\phi_w$. Now i.i.d.~pairs imply that $\ell_t$ is independent of $\Fc_{t-w-1}$, and
    \begin{equation}\label{eq:stationary-moments}
        \E_0 \ell_t = I_w, \quad \E_0 \ell_t^2 \le M_{2,w}.
    \end{equation}
    
    Now define the raw sum and its crossing time by
    \[
        U_w = 0, \quad U_t = \sum_{s=w+1}^t \ell_s, \quad T' = \inf \{t > w: U_t \ge b\}.
    \]
    For a deterministic integer $N > w$, set
    \[ 
        \sigma_N = T' \wedge N, \quad a_N = \E_0 \sigma_N.
    \]
    For $t = w+1, \ldots, N+w$, the event $\{\sigma_N \ge t-w\}$ belongs to $\Fc_{t-w-1}$. Therefore, conditioning yields
    \begin{equation}\label{eq:delayed-wald}
    \begin{aligned}
        \E_0 \sum_{t=w+1}^{\sigma_N+w} \ell_t
        &= \sum_{t=w+1}^{N+w} \E_0[\ell_t \ind_{\{\sigma_N \ge t-w\}}]\\
        &= \sum_{t=w+1}^{N+w} \E_0[\ind_{\{\sigma_N \ge t-w\}} \E_0(\ell_t \mid \Fc_{t-w-1})]\\
        &= I_w \sum_{s=1}^N \P_0(\sigma_N \ge s)
        = I_w a_N,
    \end{aligned}
    \end{equation}
    and
    \begin{equation}\label{eq:delayed-squares}
    \begin{aligned}
        \E_0 \sum_{t=w+1}^{\sigma_N+w} \ell_t^2
        &= \sum_{t=w+1}^{N+w} \E_0[\ell_t^2 \ind_{\{\sigma_N \ge t-w\}}]\\
        &\le M_{2,w} \sum_{s=1}^N \P_0(\sigma_N \ge s)
        = M_{2,w} a_N.
    \end{aligned}
    \end{equation}
    Moreover, $\{\sigma_N < t \le \sigma_N + w\} \in \Fc_{t-1}$. Using~\eqref{eq:conditional-kl} and $\sum_{t=w+1}^{N+w} \ind_{\{\sigma_N < t \le \sigma_N + w\}} = w$ pathwise,
    \begin{equation}\label{eq:tail-window}
    \begin{aligned}
        \E_0 \sum_{t=\sigma_N+1}^{\sigma_N+w} \ell_t
        &= \sum_{t=w+1}^{N+w} \E_0[\ind_{\{\sigma_N < t \le \sigma_N + w\}} \E_0(\ell_t \mid \Fc_{t-1})]\\
        &\le I \sum_{t=w+1}^{N+w} \P_0(\sigma_N < t \le \sigma_N + w)
        = wI.
    \end{aligned}
    \end{equation}
    Subtracting this tail from~\eqref{eq:delayed-wald} yields
    \begin{equation}\label{eq:stopped-drift}
        I_w a_N \le \E_0 U_{\sigma_N} + wI.
    \end{equation}
    On $\{T' \le N\}$, the first crossing gives $U_{T'-1} < b$, so $U_{\sigma_N} \le b + \ell_{\sigma_N}^+$. On $\{T' > N\}$, $U_{\sigma_N} = U_N < b$. Consequently, with $O_N = \ell_{\sigma_N}^+ \ind_{\{T' \le N\}}$,
    \begin{gather*}\label{eq:overshoot-path}
        U_{\sigma_N} \le b + O_N,\\
        O_N^2 \le \sum_{t=w+1}^{\sigma_N} \ell_t^2 \le \sum_{t=w+1}^{\sigma_N+w} \ell_t^2.
    \end{gather*}
    Cauchy-Schwarz inequality and~\eqref{eq:delayed-squares} imply
    \[
        \E_0 O_N \le \sqrt{\E_0 O_N^2} \le \sqrt{M_{2,w} a_N}.
    \]
    Together with~\eqref{eq:stopped-drift} and $I_w \ge \underline I_w > 0$, this gives
    \begin{equation}\label{eq:quadratic-start}
        \underline I_w a_N \le b + wI + \sqrt{M_{2,w} a_N}.
    \end{equation}
    Solving this quadratic inequality on $\sqrt{a_N}$ gives
    \begin{align}\label{eq:quadratic-solved}
        a_N
        \le \frac{b + wI}{\underline I_w}
        +\frac{M_{2,w}}{2 \underline I_w^2}
        +\sqrt{\frac{M_{2,w}(b + wI)}{\underline I_w^3} +\frac{M_{2,w}^2}{4\underline I_w^4}}
        \le \frac{b + wI}{\underline I_w}
        +\frac{M_{2,w}}{\underline I_w^2}
        +\sqrt{\frac{M_{2,w}(b + wI)}{\underline I_w^3}}.
    \end{align}
    The right-hand side is finite and independent of $N$. Since $\sigma_N \uparrow T'$, monotone convergence proves
    \begin{equation}\label{eq:raw-delay}
    \begin{aligned}
     \E_0 T'
     \le \frac{b + wI}{\underline I_w} + \frac{M_{2,w}}{\underline I_w^2} + \sqrt{\frac{M_{2,w}(b + wI)}{\underline I_w^3}} < \infty.
    \end{aligned}
    \end{equation}
    Thus $T'$ is integrable and almost surely finite.
    
    On the same observation sequence,
    \[
        S_{w+1}=U_{w+1}, \quad
        S_t-U_t = \max(S_{t-1},0) - U_{t-1} \ge 0
    \]
    whenever $S_{t-1} \ge U_{t-1}$.
    Induction gives $S_t \ge U_t$ for $t > w$, hence
    \begin{equation}\label{eq:cusum-raw}
        T_b \le T',\qquad \E_0 T_b \le \E_0 T'.
    \end{equation}
    
    The restart argument of \cite[Lemma~4]{xie2023} applies to the i.i.d.~post-change tail and the common fixed-window map in Assumption~\ref{ass:core-window}. That gives
    \begin{equation}\label{eq:wadd-equality}
        \Dc(T_b) \le \E_0 T_b.
    \end{equation}
    Together with~\eqref{eq:raw-delay} and~\eqref{eq:cusum-raw}, this proves~\eqref{eq:core-delay}.
\end{proof}

\begin{proof}[Proof of Corollary~\ref{cor:core-first-order}]

    The assumptions imply $\Ec_w \to 0$. For all sufficiently large $\gamma$,
    \[
        0 \le C_K \Ec_w \le I/2, \quad 0 < I/2 \le \underline I_w \le I.
    \]
    In particular, $\underline I_w > 0$ eventually. Define
    \[
        \alpha_w = \frac{C_K \Ec_w}{I}, \quad
        \overline M_2 = \sup_w M_{2,w} < \infty.
    \]
    Then $\alpha_w \to 0$ and $\underline I_w = I (1 - \alpha_w)$. Dividing~\eqref{eq:core-delay} by $b/I$ gives
    \begin{equation}\label{eq:normalized-upper}
        \frac{\Dc(T_b)}{b/I}
        \le \frac{1 + wI/b}{1 - \alpha_w}
        + \frac{\overline M_2}{bI(1 - \alpha_w)^2}
        + \sqrt{\frac{\overline M_2}{bI} \frac{1 + wI/b}{(1 - \alpha_w)^3}}.
    \end{equation}
    Here, the first term tends to one because $\alpha_w \to 0$ and $w = o(b)$. The second and third terms tend to zero because $I > 0$ is fixed and $b \to \infty$. Thus
    \begin{equation}\label{eq:upper-limit}
        \limsup_{\gamma \to \infty} \frac{\Dc(T_{\log \gamma})}{(\log \gamma)/I} \le 1.
    \end{equation}
        Also, Theorem~\ref{thm:core-arl} gives
        $\E_\infty T_{\log\gamma} \ge w + \gamma \ge \gamma$.

    Now we invoke the classical minimax lower bound. For the fixed $\theta_0$, Assumption~\ref{ass:core-model} gives independent observations that are identically distributed within each regime, with one-observation laws $P_{\theta_\infty}$ and $P_{\theta_0}$. The common covariate law $Q$ implies
    \[
        D_{\mathrm KL}(P_{\theta_0} \| P_{\theta_\infty}) = \E_Q D_{\mathrm KL}\{p_{\theta_0}(\cdot \mid X) \| p_{\theta_\infty} (\cdot \mid X)\}
        = I \in (0, \infty).
    \]
    Let $\Cc_\gamma = \{T: \E_\infty T \ge \gamma\}$ over adapted stopping times. For these i.i.d.~observation laws, minimax information bound of \cite{lorden1971} gives
    \begin{equation}\label{eq:minimax-lower}
        \liminf_{\gamma \to \infty} \inf_{T \in \Cc_\gamma} \frac{\Dc(T)}{(\log \gamma) / I} \ge 1.
    \end{equation}
    This bound applies to all such stopping times, including the VI-based rule of VI-WLCUSUM.

    Theorem~\ref{thm:core-arl} places $T_{\log \gamma}$ in $\Cc_\gamma$. Combining~\eqref{eq:upper-limit} and~\eqref{eq:minimax-lower} proves the claim.
\end{proof}

\end{document}